\documentclass[letterpaper, 10 pt, conference]{ieeeconf}  % Comment this line out if you need a4paper

\IEEEoverridecommandlockouts                              % This command is only needed if 
\usepackage{cite}
\usepackage{amsmath,amssymb,amsfonts}
\usepackage{graphicx}
\usepackage{textcomp}

\usepackage[utf8]{inputenc}  
\usepackage{leftidx}  
\usepackage{mathtools}
\mathtoolsset{showonlyrefs,showmanualtags}
\usepackage{amssymb,amsthm,amsmath}
\usepackage{dsfont}
\usepackage{booktabs} 
\usepackage{color}
\usepackage{bbold} 
\usepackage[yyyymmdd,hhmmss]{datetime}
\usepackage{bm,upgreek}
\usepackage{tablefootnote}
\usepackage{nth}
\usepackage{pgfplots}
\usepackage{tikz}
\usetikzlibrary{shapes,snakes}
\usetikzlibrary{patterns}
\usepackage{multirow}
\usepackage{footnote}
\usetikzlibrary{hobby}
\usepackage{setspace}
\usepackage{textcomp}
\usepackage{algorithm}
\usepackage{textcomp}
\usepackage{xcolor}
\usepackage{multicol}
\usepackage{algpseudocode}
\def\BibTeX{{\rm B\kern-.05em{\sc i\kern-.025em b}\kern-.08em
   T\kern-.1667em\lower.7ex\hbox{E}\kern-.125emX}}
\newtheorem{assumption}{Assumption}
\newtheorem{theorem}{Theorem}
\newtheorem{definition}{Definition}
\newtheorem{remark}{Remark}
\newtheorem{lemma}{Lemma}

\newtheorem{problem}{Problem}

\usepackage{babel}
\usepackage{tikz}
\usetikzlibrary{graphs,graphs.standard}

\tikzset{%
  every neuron/.style={
    circle,
    draw,
    minimum size=.7cm
  },
  neuron missing/.style={
    draw=none, 
    scale=3,
    text height=0.333cm,
    execute at begin node=\color{black}$\vdots$
  },
}

\definecolor{red}{RGB}{187,0,0}
\definecolor{blue}{RGB}{0, 0,180}
\definecolor{pink}{RGB}{203, 76, 178}

\algnewcommand\algorithmicinput{\textbf{Input:}}
\algnewcommand\algorithmicoutput{\textbf{Output:}}
\algnewcommand\Input{\item[\algorithmicinput]}
\algnewcommand\Output{\item[\algorithmicoutput]}

\def\BibTeX{{\rm B\kern-.05em{\sc i\kern-.025em b}\kern-.08em
    T\kern-.1667em\lower.7ex\hbox{E}\kern-.125emX}}
\begin{document}
\title{ Risk-Constrained Control of Mean-Field Linear Quadratic Systems}
\author{Masoud Roudneshin, Saba Sanami, and Amir G. Aghdam% <-this % stops a space
%\thanks{*This work was not supported by any organization}% <-this % stops a space
\thanks{Masoud Roudneshin, Saba Sanami and Amir G. Aghdam are with the Department of Electrical and Computer Engineering, Concordia Univerity, Montreal, QC, Canada. Email: {\tt\small masoud.roudneshin.concordia.ca, amir.aghdam@concordia.ca}}
}
\pagestyle{empty} % Removes all the page numbers (except for the

\maketitle
\thispagestyle{empty} % Removes the page number in the first page
\begin{abstract}
The risk-neutral LQR controller is optimal for stochastic linear dynamical systems. However, the classical optimal controller performs inefficiently in the presence of low-probability yet statistically significant (risky) events. The present research focuses on infinite-horizon risk-constrained linear quadratic regulators in a mean-field setting. We address the risk constraint by bounding the cumulative one-stage variance of the state penalty of all players. It is shown that the optimal controller is affine in the state of each player with an additive term that controls the risk constraint. In addition, we propose a solution independent of the number of players. Finally, simulations are presented to verify the theoretical findings.
\end{abstract}

\section{Introduction}

The performance evaluation of dynamical systems in the optimal control framework has long been studied in the literature. Specifically, in the linear quadratic regulator (LQR) with noisy inputs, the focus is on minimizing the expected cumulative time-average quadratic cost, also known as a risk-neutral setting \cite{opt1}. However, such a risk-neutral framework often exhibits unsatisfactory performance in real-world control systems. For instance, there exists a rich body of research to address risk in different areas, including robotics \cite{robo1, robo2}, financial systems \cite{finance1, finance2}, power grids \cite{energy1, energy2}, and multi-agent networks \cite{multi1, multi2}. Moreover, neglecting the effect of low-probability severe external events may lead to catastrophic consequences in dynamic systems, like crashing in a flock of UAVs or an autonomous vehicle hitting other vehicles and pedestrians.

There has been an increasing interest in the research community recently in the risk assessment of dynamical systems by deriving closed-form solutions for a single-agent setting \cite{tamar cdc, pappas, basar transaction}. Specifically, by solving a set of Riccati and fixed-point equations, one can obtain an affine form of the policy to meet the system’s constraints. However, in the control of a large number of agents, such a method may not provide sufficient efficacy.

This research considers the problem of exchangeable agents (players) in a mean-field setting. In such a setting, all agents have similar dynamics, and the players' states evolve as a linear function of their previous states and the overall average state. Using the results in mean-field theory, we show that the required Riccati equation (whose size increases with the number of players) can be decomposed into two Riccati equations with the same dimension as the agents' states. Furthermore, we propose a primal-dual algorithm to solve the problem iteratively.

The rest of the paper is organized as follows. In Section~II, we present some preliminaries and formulate the problem. The solution to the optimization problem is derived in Section~III, followed by simulations to validate the results in Section~IV. Finally, some concluding remarks and directions for future research are given in Section~V.

\section{Problem Formulation}
Throughout the paper, $\mathbb{R}$, $\mathbb{R}_{>0}$ and $\mathbb{N}$ represent the sets of real, positive real and natural numbers, respectively.  Given any $ n \in \mathbb{N}$, $\mathbb{N}_n$, and $\mathbf{I}_{n \times n}$ denote the finite set $\{1,\ldots,n\}$, and the $n\times n$ identity matrix, respectively.    $\| \boldsymbol \cdot \|$ is the  spectral norm of a matrix,  $\text{Tr}(\boldsymbol\cdot)$ is the trace of a matrix,  $\tau_{\text{min}}(\boldsymbol \cdot)$ is the minimum singular value of a matrix, $\rho(\boldsymbol \cdot)$ is  the spectral radius of a matrix,   and  $\text{diag}(\Lambda_1, \Lambda_2)$ is the block diagonal matrix $[\Lambda_1\quad 0;0 \quad \Lambda_2]$, and $\text{diag}(\Lambda)_{i=1}^k$ denotes a bloack-diagonal matrix with $k$ times repetition of the matrix $\Lambda$. For vectors $x,y$ and $z$, $\text{vec}(x,y,z)=[x^\intercal, y^\intercal,z^\intercal]^\intercal$ is a column vector, $x_{1:t}$ denotes the vector $(x_1,...,x_t)$ and the operator $\otimes$ denotes the Kronecker product between two matrices of appropriate size. Also, the rectified linear function is denoted by the operator $[x]_{+}=\text{max}\{0,x\}$.

\subsection{General Form of the Problem}
Given $n\in\mathbb{N}$ players, let $x^i_t \in \mathbb{R}^{d_x}$, $u^i_t \in \mathbb{R}^{d_u}$ and $w^i_t \in \mathbb{R}^{d_x}$ denote, respectively, the state, action and local noise of player $i \in \mathbb{N}_n$ at time $t \in \mathbb{N}$, where $d_x,d_u \in \mathbb{N}$. Define  the  mean-state of the players as $\bar x_t:=\frac{1}{n}\sum_{i=1}^n  x^i_t$.
 The initial states $\{x^1_0,\ldots,x^n_0 \}$ are  random  with   finite covariance matrices.   The evolution of  the state of any player $i \in \mathbb{N}_n$ at time $t \in \mathbb{N}$  is given by:
\begin{equation}\label{eq:dynamics_original}
x^i_{t+1}=Ax^i_t +Bu^i_t+ \bar A \bar x_t+ \bar B \bar u_t +w^i_t,
\end{equation}
where $\{w^i_t\}_{t=0}^\infty$ is an independent and identically distributed (i.i.d.) zero-mean  noise process with  a finite  covariance matrix.

The per-step cost of all players at time $t \in \mathbb{N}$ is given by:
\begin{equation}\label{eq:cost_original}
c_t=(\bar{x}_t)^{\intercal}\bar Q\bar{x}_t+(\bar{u}_t)^{\intercal}\bar R\bar{u}_t +\frac{1}{n} \sum_{i=1}^n(x_t^i)^{\intercal}Q x_t^i +(u_t^i)^{\intercal}R u_t^i,
\end{equation}
where $Q$, $\bar{Q}$, $R$, and $\bar{R}$ are symmetric matrices with appropriate dimensions. 
\begin{definition}
Let $h_t^i = \{x^i_{0}, u^i_{0}, ...,x^i_{t-1}, u^i_{t-1},x^i_{t}\}$ denote the history trajectory of player $i\in\mathbb{N}_{n}$. Then, the per-step risk factor for the $i$th player is defined as
\begin{equation}\label{risk player}
    d_t^i = \Big((x_t^i)^{\intercal}Q x_t^i - \mathbb{E}{\big[ (\mathbf x_t^i)^{\intercal}Q \mathbf x_t^i|h_t^i\big]} \Big)^2. 
\end{equation}
\end{definition}

\begin{assumption}\label{ass: existence, my pr}
It is assumed hereafter that the pair $(A, B)$ is stabilizable, the pair $(A,Q^{\frac{1}{2}})$ is detectable, and matrices $Q$ and $R$ are positive semi-definite and positive definite, respectively.
\end{assumption}
\begin{assumption}\label{same noise stat}
The local noises $w_t^1,...,w_t^n$ have the same distribution. 
\end{assumption}
\begin{assumption}
The noise $w_t^i$ for every player $i\in\mathbb{N}_n$ has a finite fourth-order moment, i.e., $\mathbb{E}\lVert w_t^i\rVert^4<\infty.$
\end{assumption}

In this paper, we consider the infinite-horizon risk-constrained LQR for a team of cooperative players to minimize a common cost. Also, it is desired to constrain the cumulative per-step risk of all players. This leads to the following constrained optimization problem
\begin{subequations}\label{cooperative}
\begin{align}
\label{optfunc}\mathrm{minimize}& \hspace{0.5cm}J = \limsup_{T\rightarrow \infty }\frac{1}{T}\mathbb{E}{\bigg[\sum_{t = 0}^{T} c_t \bigg]}\\\label{dynconst}
\mathrm{s.t.}&\hspace{0.5cm}\eqref{eq:dynamics_original} \hspace{0.25cm}\text{and } \forall i\in \mathbb{N}_n, \\\label{optconst}
 &\hspace{0.5cm} J_c= \frac{1}{n}\sum_{i = 1}^{n}\limsup_{T\rightarrow \infty }\frac{1}{T}\mathbb{E}{\bigg[\sum_{t = 0}^{T} d_t^i \bigg]}\leq \Gamma, 
\end{align}
\end{subequations}
where $\Gamma>0$ is a predefined risk tolerance of the user. 
\begin{remark}
From~\cite{masoud cdc, Jalal1}, when player $i \in \mathbb{N}_n$ at any time $t \in \mathbb{N}$ observes its local state $x^i_t$ and  the mean state~$\bar x_t$, i.e. $\{x^i_{1:t},\bar x_{1:t}\}$, an information structure called \textit{deep state sharing} (DSS) is considered. 
\end{remark}
\begin{definition}
Let the control input of player $i \in \mathbb{N}_n$ at time $t$ be denoted by $u^i_t=\phi^i_t(x^i_{1:t},\bar x_{1:t})$. Define  ${\Phi}^i:=\{\phi^i_t\}_{t=1}^\infty$ and $\boldsymbol{\Phi}_n:=\{ \boldsymbol {\Phi}^1,\ldots,\boldsymbol{\Phi}^n\}$ as the control strategy of player $i$ and that of all players, respectively.
\end{definition}
We now present the main problem of this article.

\begin{problem}\label{problem 1}
Consider the risk-constrained mean-field LQR problem in \eqref{cooperative}. Given the system dynamics \eqref{eq:dynamics_original}, find an optimal control strategy $\boldsymbol{\Phi}^*$ such that for any arbitrary control law $\boldsymbol{\Phi}$, the cost function \eqref{optfunc} under the constraints \eqref{dynconst} and \eqref{optconst} satisfies the following inequality
\begin{equation*}
    J(\boldsymbol{\Phi}^*) \leq J(\boldsymbol{\Phi}).
\end{equation*}
\end{problem}

\section{Main Results}
In this section, we propose a step by step solution to the optimization problem~\eqref{cooperative}.  
\subsection{Problem Reformulation}

Define a new transformed state $\Tilde{x}_t^i = x_t^i-\bar{x}_t$ for player $i \in \mathbb{N}_n$. Define also the mean control input of all players as $\bar u_t:=\frac{1}{n}\sum_{i=1}^n u^i_t$, and the transformed control input of player $i \in \mathbb{N}_n$ as $\Tilde{u}_t^i = u_t^i-\bar{u}_t$. It follows from \cite{Jalal1} that
\begin{equation}\label{dynammics transformed}
\begin{split}
  \Tilde{x}_{t+1}^i &= A \Tilde{x}_t^i  + B \Tilde{u}_t^i + \Tilde{w}_t^i\\
   \bar{x}_{t+1} &= \mathcal{A} \bar{x}_t  + \mathcal{B}  \bar{u}_t + \bar{w}_t,
\end{split}
\end{equation}
where $\mathcal{A} = A + \bar{A}$, $\mathcal{B} = B + \bar{B}$, $\bar w_t:=\frac{1}{n}\sum_{i=1}^n w^i_t$ and $\Tilde{w}_t^i = w_t^i-\bar{w}_t$.

Next, define the first and second-order moments (mean and covariance) of each player's local noise as $m_1=\mathbb{E}[{w}_t^i]$ and $M_2 = \mathbb{E}[({w}_t^i - m_1^i)({w}_t^i - m_1^i)^\intercal]$, respectively. Furthermore, let the next two higher order moments of the local noise be defined as
\begin{equation}
    \begin{split}
    {M}_3&=\mathbb{E}[({w}_t^i - m_1^i)({w}_t^i -  m_1^i)^\intercal{Q}({w}_t^i - m_1^i)],\\
    {M}_4&= \mathbb{E}[({w}_t^i - m_1^i)^\intercal{Q}({w}_t^i - m_1^i) - \text{Tr} (M_2{Q})]^2.\\
    \end{split}
\end{equation}
Also, for future reference, define $\mathbb{m}_1 = \mathbb{E}[\Tilde{w}_t^i]$ and  $\mathbb{M}_1 = \mathbb{E}[(\Tilde{w}_t^i - \mathbb{M}_1)(\Tilde{w}_t^i - \mathbb{m}_1)^\intercal]$.
\begin{lemma}\label{risk decomposed}
The risk-constrained optimization problem in \eqref{cooperative} can be reformulated as
\begin{equation}\label{cooperative-tractable}
\begin{split}
%\underset{\boldsymbol \Theta}{\mathrm{minimize}}
     &{\mathrm{minimize}} \hspace{1cm}J = \limsup_{T\rightarrow \infty }\frac{1}{T}\mathbb{E}{\bigg[\sum_{t = 0}^{T} c_t^i \bigg]}\\
     &\mathrm{s.t.}\hspace{2cm}\eqref{dynammics transformed} \hspace{0.25cm}\mathrm{and},\\
     &\hspace{2.3cm}\Tilde{J}_c= J_{\bar c}+\sum_{i=1}^n \Tilde{J}_c^i \leq \Lambda\hspace{0.25cm} \forall i\in \mathbb{N}_n, \\
\end{split}
\end{equation}
where 
\begin{equation*}
    \begin{split}
         &{J}_{\Tilde{c}}^i= \lim_{T\rightarrow \infty }\frac{1}{T}\mathbb{E}{\bigg[\sum_{t = 0}^{T} \frac{4}{n}  (\Tilde{x}_t^i)^\intercal  Q  M_2  Q \Tilde{x}_t^i   \bigg]},\\
     &J_{\bar c}= \limsup_{T\rightarrow \infty }\frac{1}{T}\mathbb{E}{\bigg[\sum_{t = 0}^{T} 4  ( \bar{x}_t)^\intercal  Q {{M}_2} Q \bar{x}_t + 4  ( \bar{x}_t)^\intercal Q {M}_3 \bigg]},
    \end{split}
\end{equation*}
and $\Lambda = \Gamma - {m}_4 + \text{Tr} (M_2{Q})^2$. 
\end{lemma}

\begin{proof}
Using the results in \cite{pappas}, the constraint in \eqref{optconst} can be reformulated as
\begin{equation*}
    J_c = \frac{1}{n}\sum_{i = 1}^{n}\limsup_{T\rightarrow \infty }\frac{1}{T}\mathbb{E}{\sum_{t = 0}^{T}} 4( {x}_t^i)^\intercal  Q {{M}_2} Q {x}_t^i + 4  ({x}_t^i)^\intercal Q {M}_3.
\end{equation*}
The proof follows immediately by rewriting the above equation as ${x}_t^i = \Tilde{x}_t^i + \bar{x}_t$, and on noting that $\sum_{i=1}^n \Tilde{x}_t^i = 0$. 
\end{proof}
\subsection{Primal-Dual Approach}
To solve the constrained optimization problem \eqref{cooperative-tractable}, we use $\lambda\geq 0$ as the Lagrange multiplier. The Lagrangian can then be expressed as
\begin{equation}\label{Lagrangian}
\mathcal{L}(\boldsymbol{\Phi},\lambda) =J + {\lambda}({\Tilde{J}_{c}}- {\Lambda}).  
\end{equation}
\begin{definition}
Define the  matrices $Q_{c} = \frac{4}{n} Q M_2 Q$, $Q_{\bar c} = 4 Q M_2 Q$, $ Q_{\lambda} = \frac{1}{n}Q + \lambda Q_{c}$, and ${Q}_{ \bar \lambda} = Q + \bar Q + \lambda Q_{\bar c}$.
\end{definition}

\begin{lemma}\label{lagrange decompose}
The Lagrangian in \eqref{Lagrangian} can be reformulated as
\begin{equation}\label{lagrangian new}
  \mathcal{L}(\boldsymbol{\Phi},\lambda) =  \bar {\mathcal{L}} + \sum_{i=1}^n\mathcal{L}^i, 
\end{equation}
where
\begin{equation*}
    \begin{split}
        \mathcal{L}^i &= \limsup_{T\rightarrow \infty }\frac{1}{T}\mathbb{E}\bigg[(\Tilde{x}_t^i)^\intercal(Q_{\lambda})\Tilde{x}_t^i + (\Tilde{u}_t^i)^\intercal\frac{1}{n}R\Tilde{u}_t^i \bigg],\\
        \bar{\mathcal{L}} &= \limsup_{T\rightarrow \infty }\frac{1}{T}\mathbb{E}\bigg[ \bar{x}_t^\intercal({Q}_{\bar{\lambda}})\bar{x}_t + S_{\lambda}\bar{x}_t + \bar{u}_t^\intercal(R+\bar{R})\bar{u}_t\bigg].
    \end{split}
\end{equation*}
\end{lemma}
\begin{proof}
The result follows directly from Lemma~\ref{risk decomposed}, the definition of the per-step cost in \eqref{eq:cost_original}, and on noting that $\sum_{i=1}^n \Tilde{x}_t^i = 0$ and $\sum_{i=1}^n \Tilde{u}_t^i = 0$.
\end{proof}
To solve for the optimal value of the Lagrangian $\mathcal L^*$ in \eqref{Lagrangian}, we find the general form of the policies for a constant multiplier.  
\begin{theorem}\label{primal optimal policy}
 For a fixed multiplier $\lambda$, the optimal policy for each player is affine, such that
 \begin{equation}\label{Main policy}
     {u}_t^i = - \theta(\lambda){x}_t^i - (\bar \theta(\lambda) - \theta(\lambda))\bar{x}_t+ \tau (\lambda) + \bar\tau (\lambda) ,
 \end{equation}
in which 
\begin{equation}\label{optimal gain}
\begin{split}
\theta &= -({R} + {B}^\intercal{P}{B})^{-1}{B}^\intercal{P}{A},\\
\bar \theta &= -(\mathcal{R} + \mathcal{B}^\intercal\mathcal{P}\mathcal{B})^{-1}\mathcal{B}^\intercal\mathcal{P}\mathcal{A},  
\end{split}
\end{equation}
and 
\begin{equation}\label{optimal gain affine}
\begin{split}
\tau &= -\frac{1}{2}({R} + {B}^\intercal{P}{B})^{-1}{B}^\intercal(2{P}\mathbb{m}_1+{g}),\\
\bar\tau &= -\frac{1}{2}(\mathcal{R} + \mathcal{B}^\intercal\mathcal{P}\mathcal{B})^{-1}\mathcal{B}^\intercal(2\mathcal{P}{m}_1+ \mathbf{g}),
\end{split}
\end{equation}
where $P$, $\mathcal{P}$, $g$ and $\mathbf{g}$ are obtained by solving the following recursive equations
\begin{equation}\label{Riccati}
\begin{split}
 {P} &= Q_{\lambda}{A}^\intercal{P}{A} -{A}^\intercal{P}{B}({R} + {B}^\intercal{P}{B})^{-1}{B}^\intercal{P}{A},\\
 \mathcal{P} &= Q_{\bar \lambda} \mathcal{A}^\intercal \mathcal{P} \mathcal{A} - \mathcal{A}^\intercal \mathcal{P} \mathcal{B}(\mathcal{R} +  \mathcal{B}^\intercal \mathcal{P} \mathcal{B})^{-1} \mathcal{B}^\intercal \mathcal{P} \mathcal{A},\\
 {g}^\intercal &=(2\mathbb{m}_1^\intercal {P}+ {g}^\intercal)({A} -{B}\theta),\\
  \mathbf{g}^\intercal &= (2{m}_1^\intercal \mathcal{P}+ \mathbf{g}^\intercal)(\mathcal{A} - \mathcal{B}\bar \theta)+ 4\lambda (QM_3)^\intercal.  
\end{split}
\end{equation}
\end{theorem}
\begin{proof}
Define the generalized state and action of all agents in an augmented form as $\mathbf{x}_t = [\text{vec}(\Tilde{x}_t^i)_{i=1}^n,\bar{x}_t]$ and $\mathbf{u}_t = [\text{vec}(\Tilde{u}_t^i)_{i=1}^n,\bar{u}_t]$, respectively. Then, it follows that
\begin{equation*}
  \mathbf{x}_{t+1} = \mathbf{A}  \mathbf{x}_{t} +  \mathbf{B}  \mathbf{u}_{t},
\end{equation*}
where
\begin{equation*}
  \mathbf{A}= \text{diag}(\text{diag}(A)_{i=1}^n,\bar{A}),\quad \mathbf{B}= \text{diag}(\text{diag}(B)_{i=1}^n,\bar{B}).  
\end{equation*}
Define the finite-horizon Lagrangian as the value function $V_T$ and note that the results in Theorem~2 of \cite{pappas} imply that the Lagrangian has a quadratic form as 
\begin{equation*}
 V_T = \mathbf{x}_t^\intercal\mathbf{P}\mathbf{x}_t + \mathbf{g}\mathbf{x}_t + \mathbf{z}_t.
\end{equation*}
Instead of solving for the optimal policy in the larger state-space of $\mathbf{x}_t$, from Lemma~\ref{lagrange decompose}, the value function can also be decomposed into a set of smaller value functions such that
\begin{equation*}
    V_T = \bar{V}_T + \sum_{i=1}^n \Tilde{V}_T^i.
\end{equation*}
Since the Lagrangians $\bar{\mathcal{L}}$ and $\Tilde{\mathcal{L}}^i$ have complete square forms, the minimization can be carried out over the smaller state space of $\Tilde{x}_t^i$ and $\bar{x}_t$. Therefore, by employing dynamic programming, we have the following two recursive optimality equations
\begin{equation*}
\begin{split}
 \Tilde{V}_T^i &= \underset{ \Tilde{u}_t}{\mathrm{min}}\big((\Tilde{x}_t^i)^\intercal Q_{\lambda}\Tilde{x}_t^i + \frac{1}{n} (\Tilde{u}_t^i)^\intercal R\Tilde{u}_t^i +\bar{V}_{T+1}^i\big),\\
    \bar{V}_T &= \underset{ \bar{u}_t}{\mathrm{min}}(\bar{x}_t^\intercal Q_{\bar\lambda}\bar{x}_t + \bar{u}_t^\intercal({R+\bar R})\bar{u}_t +\bar{V}_{T+1}).    
\end{split}
\end{equation*}
The proof follows by taking the derivative with respect to $\Tilde{u}_t^i$ and $\bar{u}_t$ and using backward dynamic programming.
\end{proof}
\begin{remark}[Strong Duality]
Using the results established in Theorem~2 of \cite{tamar cdc} and \cite{tamar journal}, there exists an optimal multiplier $\lambda^*$ such that the policy 
 \begin{equation}
     {u}_t^i = - \theta(\lambda^*){x}_t^i - (\bar \theta(\lambda^*) - \theta(\lambda^*))\bar{x}_t+ \tau (\lambda^*) + \bar\tau (\lambda^*)
 \end{equation}
 is the optimal solution to \eqref{cooperative-tractable}.
\end{remark}

\subsection{Solution of the Dual Problem with Subgradients }
Since there is no optimality gap in the optimization problem \eqref{cooperative}, we can alternatively solve the following dual problem
\begin{equation}
   \underset{ \lambda\geq 0}{\mathrm{max}}\hspace{0.25cm} D(\lambda) =\underset{ \lambda\geq 0}{\mathrm{max}}\hspace{0.1cm} \underset{\mathbf{u}}{\mathrm{min}}\hspace{0.25cm}\mathcal{L}(\mathbf{u},\lambda)
\end{equation}
which is also concave in $\lambda$. Let $d$ denote the subgradient. Then, from the results in \cite{dual1,dual2}, the subgradient of $ D(\lambda)$ can be expressed as
\begin{equation}\label{subg vec}
 d = \Tilde{J}_c(\theta, \bar \theta,\lambda) - \Lambda.
\end{equation}
The following theorem provides the explicit form of the constraints for deriving the subgradient vector. 
\begin{theorem}\label{subg comp}
Consider the stabilizing control input given by \eqref{Main policy}. Then,
\begin{equation*}
\begin{split}
 {J}_{\Tilde{c}}^i &= \text{Tr}\bigg[P_{\tilde{c}}\big(\mathbb{M}_2+(B\tau + \mathbb{m}_1)(B\tau + \mathbb{m}_1)^\intercal\big) \bigg]\\
 &+ g_{\tilde{c}}^\intercal (B\tau + \mathbb{m}_1),\\
   {J}_{\bar{c}} &= \text{Tr}\bigg[P_{\bar{c}}(\mathcal{B}\bar\tau + {m}_1)(\mathcal{B}\bar\tau + {m}_1)^\intercal \bigg] + g_{\bar{c}}^\intercal (\mathcal{B}\bar\tau + {m}_1),
\end{split}
\end{equation*}
where $P_{\bar c}$ and $P_c$ are the positive definite solutions of the following Lyapunov equations
\begin{equation}\label{Lyp const}
\begin{split}
{P}_{\tilde{c}} &= \frac{4}{n} Q{M}_2 Q+ (A - B\theta)^\intercal {P}_{\tilde{c}}(A - B\theta),\\
P_{\bar c} &= 4Q {{M}_2} Q+(\mathcal{A}-\mathcal{B}\bar{\theta})^\intercal P_{\bar c}(\mathcal{A}-\mathcal{B}\bar{\theta}), \\
\end{split}
\end{equation}
where
\begin{equation*}
\begin{split}
 g_{\tilde{c}}^\intercal &= 2 \big\{({B}{\tau} +\mathbb{m}_1)^\intercal P_{\tilde{c}}({A}-{B}{\theta})\big\} \big(I- {A}+{B}{\theta}\big)^{-1},\\
g_{\bar c}^\intercal &= 2 \big\{(\mathcal{B}\bar{\tau} +{m}_1)^\intercal P_{\bar c}(\mathcal{A}-\mathcal{B}\bar{\theta})+2M_3^\intercal Q\big\} \big(I- \mathcal{A}+\mathcal{B}\bar{\theta}\big)^{-1}.
\end{split}
\end{equation*}
\end{theorem}

\begin{proof}
Define the relative value functions 
\begin{equation}
\begin{split}
   {V}_{\tilde{c}}^i &= \mathbb{E} \bigg[\sum_{t = 0}^{\infty} \frac{4}{n}  (\Tilde{x}_t^i)^\intercal  Q  M_2  Q \Tilde{x}_t^i - {J}_{\Tilde{c}}^i   \bigg],\\
   {V}_{\bar{c}} &= \mathbb{E} \bigg[\sum_{t = 0}^{\infty} 4  ( \bar{x}_t)^\intercal  Q {{M}_2} Q \bar{x}_t + 4  ( \bar{x}_t)^\intercal Q {M}_3 - {J}_{\bar{c}} \bigg].
\end{split}
\end{equation}
Using backward dynamic programming, it can be shown that such value functions have a quadratic form, i.e. ${V}_{\tilde{c}}^i =  (\Tilde{x}_t^i)^\intercal P_{\tilde{c}} \Tilde{x}_t^i + g_{\tilde{c}}^\intercal \Tilde{x}_t^i + z_{\tilde{c}}$ and ${V}_{\bar{c}}^i =  \bar{x}_t^\intercal P_{\bar{c}} \bar{x}_t + g_{\bar{c}}^\intercal \bar{x}_t + z_{\bar{c}}$. Using the Bellman equation for ${V}_{\tilde{c}}^i$ one has
\begin{equation*}
  \begin{split}
V_{\tilde{c}}^i=&(\Tilde{x}_t^i)^\intercal P_{\tilde{c}} \Tilde{x}_t^i + g_{\tilde{c}}^\intercal \Tilde{x}_t^i + z_{\tilde{c}} \\
=& \frac{4}{n}  (\Tilde{x}_t^i)^\intercal  Q  M_2  Q \Tilde{x}_t^i -{J}_{\Tilde{c}}^i +\mathbb{E}[g_{\tilde{c}}^\intercal\big((A-B\theta)\Tilde{x}_t^i + B\tau +\Tilde{w}_t^i\big)] + z_{\tilde{c}}\\
&+ \mathbb{E}[(A-B\theta)\Tilde{x}_t^i + B\tau +\Tilde{w}_t^i]^\intercal P_{\tilde{c}} [(A-B\theta)\Tilde{x}_t^i + B\tau +\Tilde{w}_t^i]\\
=&(\Tilde{x}_t^i)^\intercal\bigg[ \frac{4}{n} Q  M_2  Q + (A-B\theta)^\intercal P_{\tilde{c}} (A-B\theta) \bigg] (\Tilde{x}_t^i)\\
&\bigg[ 2(B{\tau} +{m}_1)^\intercal P_{\tilde{c}}(A-B\theta) + g_{\tilde{c}}^\intercal(A-B\theta)\bigg](\Tilde{x}_t^i)  - {J}_{\Tilde{c}}^i + z_{\tilde{c}} + \\
&\text{Tr}\bigg[P_{\tilde{c}}(\mathbb{M}_2+(B\tau + \mathbb{m}_1)(\mathbb{M}_2+(B\tau + \mathbb{m}_1)^\intercal \bigg] + g_{\tilde{c}}^\intercal (B\tau + \mathbb{m}_1),
  \end{split} 
\end{equation*}
Then, it follows that
\begin{equation*}
    {J}_{\Tilde{c}}^i = \text{Tr}\bigg[P_{\tilde{c}}(\mathbb{M}_2+(B\tau + \mathbb{m}_1)(\mathbb{M}_2+(B\tau + \mathbb{m}_1)^\intercal \bigg] + g_{\tilde{c}}^\intercal (B\tau + \mathbb{m}_1). 
\end{equation*}
Using a similar argument, ${V}_{\bar{c}}$ can be written as
\begin{equation*}
  \begin{split}
{V}_{\bar{c}} =&\bar{x}_t^\intercal P_{\bar{c}} \bar{x}_t + g_{\bar{c}}^\intercal \bar{x}_t\\
=& 4\bar{x}_t^\intercal  Q  M_2  Q \bar{x}_t + 4 M_3^\intercal  Q \bar{x}_t -{J}_{\bar{c}} +\mathbb{E}[g_{\bar{c}}^\intercal\big((\mathcal{A}-\mathcal{B}\bar\theta)\bar{x}_t + \mathcal{B}\tau +\bar{w}_t\big)]\\
&+ \mathbb{E}[(\mathcal{A}-\mathcal{B}\bar\theta)\bar{x}_t + \mathcal{B}\tau +\bar{w}_t]^\intercal P_{\bar{c}} [(\mathcal{A}-\mathcal{B}\bar\theta))\bar{x}_t +\mathcal{B}\tau +\bar{w}_t]  + z_{\bar{c}} \\
=&\bar{x}_t^\intercal\bigg[ Q  M_2  Q + (\mathcal{A}-\mathcal{B}\bar\theta)^\intercal P_{\tilde{c}} (\mathcal{A}-\mathcal{B}\bar\theta) \bigg] \bar{x}_t\\
&\bigg[ 2(\mathcal{B}{\tau} +{m}_1)^\intercal P_{\tilde{c}}(\mathcal{A}-\mathcal{B}\bar\theta) + 4 M_3^\intercal  Q+ g_{\bar{c}}^\intercal(\mathcal{A}-\mathcal{B}\bar\theta)\bigg]\bar{x}_t\\
&+\text{Tr}\bigg[P_{\bar{c}}(\mathcal{B}\tau + {m}_1)(\mathcal{B}\tau + {m}_1)^\intercal \bigg] + g_{\tilde{c}}^\intercal (\mathcal{B}\tau + {m}_1)  - {J}_{\bar{c}} + z_{\bar{c}},
  \end{split} 
\end{equation*}
which yields
\begin{equation*}
     {J}_{\bar{c}} = \text{Tr}\bigg[P_{\bar{c}}(\mathcal{B}\bar\tau + {m}_1)(\mathcal{B}\bar\tau + {m}_1)^\intercal \bigg] + g_{\bar{c}}^\intercal (\mathcal{B}\bar\tau + {m}_1). 
\end{equation*}

\end{proof}
From Theorem~\ref{subg comp}, we can compute $\Tilde{J}_c= J_{\bar c}+\sum_{i=1}^n \Tilde{J}_c^i$ and then find the subgradients accordingly. Algorithm~1 describes the proposed primal-dual method to solve the optimization problem in \eqref{cooperative}.
\begin{algorithm}[htp]
\caption{Primal-Dual Algorithm for Risk-Constrained Mean-field LQR}\label{optimalSolAlg}

\begin{algorithmic}[1]
  \Input{Initial  ${\lambda}_0$, step size $\eta$}
  \State Iteration counter $k$
  \For {$k = 1, 2,...$}
 
  \State Obtain ${u}_t = \text{argmin}\hspace{0.1cm}\mathcal{L}(\mathbf{u}_t,\boldsymbol{\lambda}_k)$ from Theorem~\ref{primal optimal policy}
  \State Compute ${d}_k$ from Theorem~\ref{subg comp}
  \State Update the multiplier $\lambda_{k+1} = [\lambda_{k} + {\eta}_{k}.{d}_k]_{+}$
\EndFor
\end{algorithmic}

\end{algorithm}

\begin{remark}
     Since the policy in \eqref{Main policy} is stabilizable, the subgradients' and multipliers' vectors have upper bounds.
\end{remark}
\begin{remark}
     Since the subgradients and multipliers are upper bounded, using an argument analogous to that in Theorem~3 in \cite{tamar cdc}, Algorithm~1 converges to the optimal policy after sufficient iterations. 
\end{remark}

\section{Simulations}
We validate the proposed method using numerical simulations on a low-inertia microgrid (MG) system. Consider the load frequency problem (LFC) with risk constraints on the agents' frequency and mean state. The MGs exchange information with each other through the mean state of the system.

\begin{table}[!b]
\centering
\caption{Simulation Parameters from \cite{microgrid}}
    \begin{tabular}{ c  c  c c}
   \hline\hline
 Damping Factor & $D$ & 16.66 & MW/Hz \\ 
 \hline
 Speed Droop & $R$ & 1.2$\times 10^{-3}$ & Hz/MW \\
 \hline
 Turbine Static Gain & $K_t$ & 1 & MW/MW \\
 \hline
 Turbine Time Constant & $T_t$ & 0.3 & s \\
 \hline
 Area Static Gain & $K_p$ & 0.06 & Hz/MW \\ %[1ex] 
 \hline
 Area Time Constant & $T_p$ & 24 & s \\ %[1ex] 
 \hline
 Tie-line Coefficient & $K_{\text{tie}}$ & 850 & MW/Hz \\ %[1ex] 
 \hline
    \end{tabular}
    
\end{table}
Consider microgrids in $n$ areas. Let $\Delta P_{\text{tie},i}$ and $\Delta f_i$ denote the power inflow and the frequency deviation corresponding to the $i$th microgrid. We assume that this power flow is proportional to the discrepancy between the frequency deviation of each area and the mean frequency deviation of all areas, i.e.
\begin{equation*}
 \Delta P_{\text{tie},i} = \int K_{\text{tie},i}(\Delta f_i - \Delta \bar f) dt
\end{equation*}
In addition, the control signal of the $i$th area is the sum of two terms given below
\begin{equation*}
    \Delta u_{tot,i} = \Delta P_{f,i} + \Delta P_{C,i},
\end{equation*}
where $\Delta P_{f,i} = -\frac{1}{R_i}\Delta f_i$, and $\Delta P_{C,i}$ denotes the automatic generation control (AGC). These two controls specify the output power of the microgrid at the $i$th area denoted by $\Delta P_{G,i}$. The other state variable is the area control error (ACE) denoted by $z_i:= \beta \Delta f_i + \Delta P_{\text{tie},i}$ with the bias factor $\beta_i = D_i + \frac{1}{R_i}$. 

The overall state of each microgrid is
\begin{equation*}
 x^i = [\Delta f_i,\Delta P_{G,i},\Delta P_{\text{tie},i},\int z_i]^\intercal.   
\end{equation*}
The dynamics of the system is 
\begin{equation*}
    x^i_{t+1} = A x^i_{t} + \bar A \bar x_t + B u_t^i,
\end{equation*}
where 
\begin{equation*}
\begin{split}
   A&=\left[\begin{array}{c c c c} 
  -\frac{1}{T_p}&\frac{K_p}{T_p}&-\frac{K_p}{T_p}&0 \\ 
  %\hline 
 -\frac{K_t}{RT_t}&-\frac{1}{T_t}&0&0\\
 0&0&0&1\\
 \beta&0&1&0
\end{array} \right]\\
\bar A &=\left[\begin{array}{c c c c} 
  0&0&0&0 \\ 
  %\hline 
 0&0&0&0\\
 K_{tie}&0&0&0\\
 \beta&0&0&0
\end{array} \right]\\
B &=\left[\begin{array}{c} 
  0 \\ 
  %\hline 
 0\\
\frac{K_t}{T_t}\\
 0
\end{array} \right]
\end{split}
\end{equation*}

We use the parameters in Table~I from \cite{microgrid}. Also, we select $Q =\text{diag}(800,80,80,4000)$, $R = 5$ and $\Lambda = 100$. Fig.~\ref{constviol} shows the constraint violation, where it is observed that as the number of iterations grows, the constraint violation tends to zero. In other words, the control law resulting from the algorithm minimizes the common cost function of players while not violating the system's constraint.{Fig.~\ref{stateComp} illustrates the variation of the first state, i.e. $\Delta f_i$, with high-amplitude disturbances at different time instants. We compare the performance of the proposed method with that of the risk-neutral control approach. It is observed that our method results in less state fluctuations and smaller overshoot in the presence of high-amplitude disturbances, confirming the results developed in Theorems~\ref{primal optimal policy} and \ref{subg comp}.}

\begin{figure}
\centering
	\includegraphics[scale=0.6]{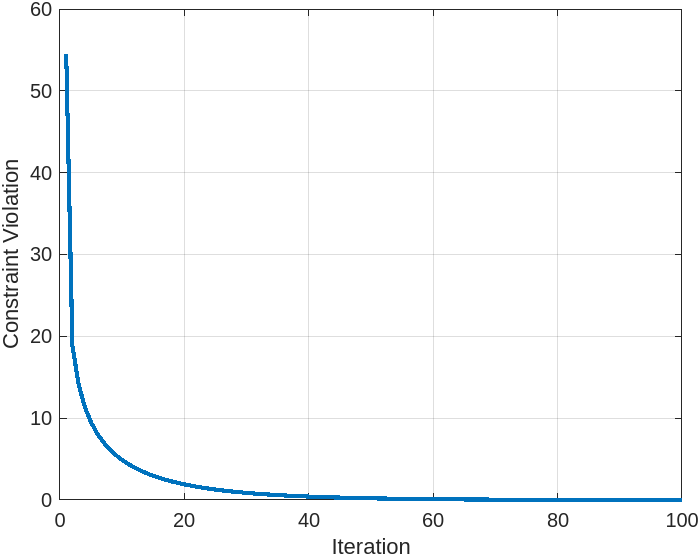}
	\caption{ Constraint violation with iterations for the microgrid problem }
	\label{constviol}
\end{figure}

\begin{figure}
\centering
	\includegraphics[scale=0.4]{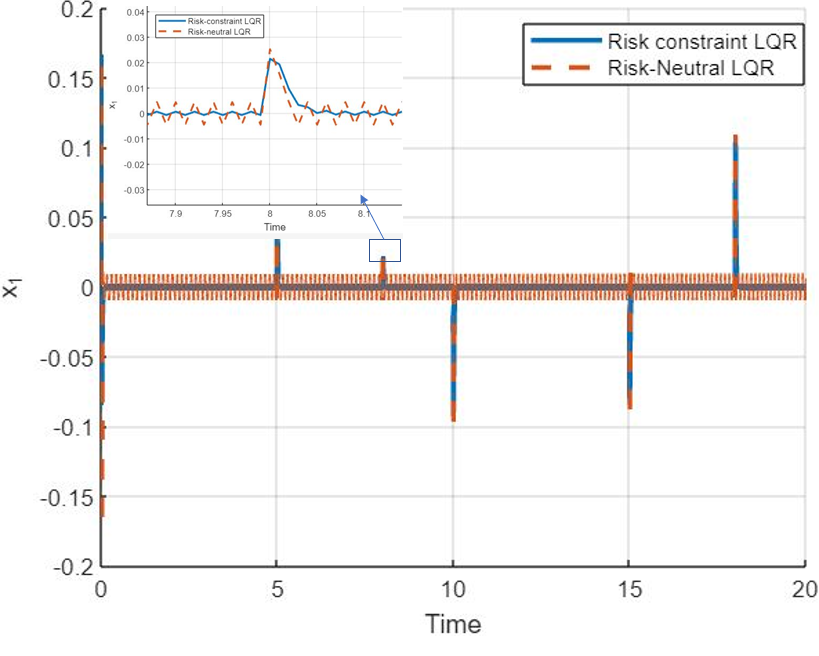}
	\caption{ {Comparison of system state using the risk-neutral controller and the proposed one}}
	\label{stateComp}
\end{figure}
\section{Conclusions}
We proposed a computationally-efficient method to tackle the problem of risk-constrained control of mean-field linear quadratic systems. The method only requires the solution of two Riccati equations and is independent of the number of players. This is a feature that is essential in controlling a multi-agent system of large size. The application of policy gradient methods as an alternative approach and considering individual constraints for the players are two interesting topics for the extension of the current research.

\end{document}